\documentclass[a4paper,11pt]{llncs}
\usepackage[margin=1in]{geometry}
\usepackage{fullpage,amsfonts,mathpazo,amssymb,amsmath,wrapfig,graphicx,tikz}
\usepackage{txfonts}
\usepackage{xfrac}
\usepackage{multirow}
\usepackage{setspace}
\usepackage{lipsum}
\usepackage{dsfont}
\usepackage[utf8]{inputenc} 
\usepackage{subfig}
\usepackage[colorinlistoftodos,prependcaption,textsize=tiny]{todonotes}
\usepackage{enumerate}
\usepackage{algorithm}
\newtheorem{result}{Result}

\newcommand{\R}{\mathbb{R}}

\newcommand{\pref}[1]{\prec_{#1}}

\newcommand{\SMP}[1]{(#1,\prec)}

\begin{document}
	
	\title{The weighted stable matching problem}
	
	\author{Linda Farczadi  \thanks{This research was supported by the The Swiss National Science Foundation, with grant ref. FNS 513835}, Nat\'alia Guri\v{c}anov\'a}

\institute{\'Ecole polytechnique f\'ed\'erale de Lausanne}

	\maketitle

	\begin{abstract}  We study the stable matching problem in non-bipartite graphs with incomplete but strict preference lists, where the edges have weights and the goal is to compute a stable matching of minimum or maximum weight. This problem is known to be NP-hard in general. Our contribution is two fold: a polyhedral characterization and an approximation algorithm. Previously Chen et al. have shown that the stable matching polytope is integral if and only if the subgraph obtained after running phase one of Irving's algorithm is bipartite. We improve upon this result by showing that there are instances where this subgraph might not be bipartite but one can further eliminate some edges and arrive at a bipartite subgraph. Our elimination procedure ensures that the set of stable matchings remains the same, and thus the stable matching polytope of the final subgraph contains the incidence vectors of all stable matchings of our original graph. This allows us to characterize a larger class of instances for which the weighted stable matching problem is polynomial-time solvable. We also show that our edge elimination procedure is best possible, meaning that if the subgraph we arrive at is not bipartite, then there is no bipartite subgraph that has the same set of stable matchings as the original graph. We complement these results with a $2$-approximation algorithm for the minimum weight stable matching problem for instances where each agent has at most two possible partners in any stable matching. This is the first approximation result for any class of instances with  general weights.
\end{abstract}
	
\section{Introduction}\label{section:introduction}
	
An instance of the \emph{Stable Matching Problem} (SMP) is a pair $\SMP G$, where $G = (V,E)$ is a graph and $\prec \; = \{\prec_v\}_{v \in V}$ is a set of \emph{preference lists}, where for each $v \in V$, $\prec_v$ is a strict linear order on $\delta(v)$, the neighbors of $v$ in $G$. The vertices of $G$ represent the set of agents and the edges correspond to mutually acceptable pairs. If an agent $v \in V$ has two neighbors $u,w$ in $G$ such that $u \pref v w$ then we say that $v$ \emph{prefers} $u$ \emph{over} $w$.  Let $M \subseteq E$ be a matching of $G$. For each vertex $v \in V$, we denote $M(v)$ its partner in $M$, or let $M(v) = v$ if $v$ is unmatched in $M$. The edge $uv \in E$ is a \emph{blocking edge} for $M$ if $u$ is unmatched or prefers $v$ to its current partner and if, at the same time, $v$ is unmatched or prefers $u$ to its current partner. A matching $M$ is \emph{stable} if no edge $e \in E$ is a blocking edge for $M$. In the \emph{Weighted Stable Matching Problem} (wSMP), in addition to $\SMP G$, a weight $w(e) \geq 0$ is assigned to each edge $e \in E$ and we wish to find a stable matching $M$ of maximum or minimum weight. 

When $G$ is bipartite, we obtain the well-studied Stable Marriage Problem first introduced by Gale and Shapley \cite{gale1962college}. In their seminal work, Gale and Shapley showed that every instance of SMP admits a solution and such a solution can be computed efficiently using the so-called deferred acceptance algorithm. In fact, both SMP and wSMP are solvable in polynomial time. In particular, there exists a system of linear inequalities, known as the stable matching polytope, that describes the convex hull of the incidence vectors of stable matchings. Since its introduction, the stable marriage problem has become one of the most popular combinatorial problems with several books being dedicated to its study \cite{gusfield1989stable}, \cite{roth1992two} and more recently \cite{manlove2013algorithmics}.  The popularity of this model arises not only from its nice theoretical properties but also from its many applications. In particular, a wide array of allocation problems from many diverse fields can be analysed within its context. Some well known examples include the labour market for medical interns, auction markets, the college admissions market and the organ donor-recipient pair market \cite{roth1992two}.

If $G$ is non-bipartite we have what is known as the Stable Roommates Problem. This problem was also proposed by Gale and Shapley \cite{gale1962college}, however its properties are quite different from the bipartite case. To begin with, not every instance admits a stable matching. Irving provided a polynomial time algorithm that either finds a stable matching or reports that none exist. Unlike in the bipartite case, the stable matching polytope is no longer integral, and the wSMP problem becomes NP-hard in general. 

\subsection{Our contribution and results.} We study instances where $G$ is non-bipartite. Our main result is to characterize a new and larger class of instances for which the wSMP can be solved in polynomial time. We say that an instance $\SMP G$ is \emph{bipartite reducible} if there exists a subgraph $H$ of $G$ such that $H$ is bipartite and has the same set of stable matchings as $G$. Our first result is polynomial time edge-elimination procedure that either finds such a bipartite subgraph or determines that the original instance is not bipartite reducible. 
\begin{result}
There exists a polynomial time algorithm that either finds a  bipartite subgraph $H$ of $G$ whose set of stable matchings is the same as that of $G$, or determines that no such subgraph exists.
\end{result}
Since the stable matching polytope is integral for bipartite graphs, this implies that we can optimize a linear function over the set of stable matchings of any bipartite reducible instance. Previously, this was known only for instances where the subgraph obtained after running phase one of Irving's algorithm is bipartite. We show that our result is a strict generalization.
\begin{result}wSMP is polynomially solvable for all instances $\SMP G$ that are bipartite reducible. Moreover, the class of bipartite reducible instances is a strict superset of the class of instances for which phase one of Irving's produces a bipartite graph.\end{result}
We then consider approximation algorithms for the minimum weight stable matching problem. Previous research on this topic has focused on special classes of weight functions such as the egalitarian stable matching and $U$-shaped weights. However, no approximation algorithm was previously known for any class of instances under general weights. We provide a first result of this kind by considering instances where each agent is matched to one of two possible partners in any stable matchings. 
\begin{result}There exists a $2$-approximation algorithm for the minimum-weight stable matching problem for instances where each agent has at most two possible partners in any stable matching.\end{result}

\subsection{Related work.} 

 Gale and Shapley \cite{gale1962} showed that a stable matching always exists if $G$ is bipartite, and gave a polynomial time algorithm known as the deferred acceptance algorithm for finding such a matching. They also observed that if $G$ is non-bipartite, a stable matching does not always exist. Irving \cite{irving1985} gave the first polynomial-time algorithm that either finds a stable matching when $G$ is non-bipartite, or determines that no such matching exists. His algorithm was originally for instances where the underlying graph $G$ is the complete graph, however it can be easily generalized for any graph $G$. 
  
 Vande Vate \cite{vandevate1992} first characterized incidence vectors of perfect stable matchings as vertices of a certain polytope in the case when $G$ is a complete bipartite graph, thus showing that wSMP is polynomial-time solvable in this case. Later, Rothblum \cite{rothblum1992} extended this result to general bipartite graphs.
	
Feder \cite{feder1992} showed that wSMP is NP-hard in general. A particular case of weight function is the \emph{egalitarian weight function}, where for each edge $e = uv$ its weight $w(e)$ is given by the sum of ranks of this particular edge in its endpoints' preference lists. Feder \cite{feder1992},\cite{feder1994} gave a 2-approximation for finding a minimum-weight stable matching with egalitarian weights and showed that there exists an $\alpha$-approximation for minimum-weight SMP if and only if there exists an $\alpha$-approximation for the Minimum Vertex Cover. Later, it was showed that, assuming the Unique Games Conjecture, Minimum Vertex Cover cannot be approximated within $2 - \epsilon$ for any $\epsilon > 0$ \cite{khot2008}. 
	
Teo and Sethuraman \cite{teo1998} constructed a 2-aproximation algorithm for the minimum-weight SMP for a larger class of weight functions than the egalitarian weight functions, namely weight functions where $w(uv) = w_u(v) + w_v(u)$ for each $uv \in E$, and the functions $\{w_u\}$ satisfy the so-called \emph{U-shape} condition. Recently, Cseh et al. \cite{cseh2016} considered a class of minimum-weight SMP problems with egalitarian weights where all preference lists are of length at most $d$. For $d = 2$, they gave a polynomial time algorithm for solving this problem while they showed that it is NP-hard even for $d = 3.$ Moreover, for $d \in\{3,4,5\}$, they gave a $\frac{2d + 3}{7}$-approximation algorithm for minimum-weight SMP with egalitarian weights, improving the results of Feder.
	
\section{Preliminaries}\label{section:preliminaries}
	
Let $\SMP G$ be an instance of SMP. If $e = uv$ and $f = wv$ are two edges in $E$ with $u \pref v w$, we can write this equivalently as $e \pref v f$, and we say that $e$ \emph{dominates} $f$ \emph{at} $v$. For each $e \in E$, let $\phi(e)$ be the set containing $e$ and all the edges that dominate $e$ at one of the endpoints. That is,
\begin{align*}
\phi(e) = \{f \in E: \exists \, v \in V: f \pref v e\} \cup \{e\}
\end{align*}
Note that under this definition, a matching is stable if and only if $|M \cap \phi(e)| \geq 1$ for each $e \in E$. For each $u,v$ s.t. $uv \in E$, we let $rank_u(v)$ be the rank of $v$ in $u$'s preference list. We say that a matching $M$ is \emph{perfect} if every agent is matched in $M$; that is $M(v) \neq v$, for all $v \in V$. For each subgraph $K$ of $G$, we let $\prec^K$ be the restriction of $\prec$ to $K$. For simplicity, we will denote the SMP instance $(K, \prec^K)$ by $(K, \prec)$. For each $v \in V$, we denote the most and least preferred partner of $v$ in $K$ by $f_K(v)$ and $l_K(v)$, respectively. We denote $E(K)$ the set of edges of $K$ and $V(K) $ the set of vertices of $K$.

\begin{definition}[Stable matching polytope]
	For an instance $\SMP G$ of SMP, we let the \emph{Stable matching polytope} $SM \SMP G$ be the convex hull of incidence vectors of its stable matchings. 
\end{definition}
If $x = \chi^M$ is an incidence vector of a stable matching $M$, it satisfies the following inequalities:
	\begin{itemize}
		\item[(i)] matching:  
		\begin{equation}
		x(\delta(v)) \leq 1 \; \forall v \in V
		\label{ineq:match}
		\end{equation} 
		\item[(ii)] stability: 
		\begin{equation}
		x(\phi(e)) \geq 1 \; \forall e \in E
		\label{ineq:stab}
		\end{equation}
		\item[(iii)] non-negativity: 
		\begin{equation}
		x \geq 0
		\label{ineq:nonneg}
		\end{equation} 

	\end{itemize}
\begin{definition}[Fractional stable matching polytope]
	For an instance $\SMP G$ of SMP the \emph{Fractional stable matching polytope} is given by
		\begin{eqnarray}
		FSM \SMP G = \left\{ x \in \R^{E} : \begin{array}{rr} x(\delta(v)) \leq 1 & \forall v \in V\\ x (\phi(e)) \geq 1 & \forall e \in E\\ x \geq 0 &  \end{array}\right\},
		\label{def:fsm}
		\end{eqnarray}
\end{definition}
The Fractional stable matching polytope was first studied by Abeledo and Rothblum \cite{abeledo1994}, who proved that the incidence vectors of stable matchings are precisely the integral vertices of $FSM \SMP G$. They also showed that this polytope is integral if $G$ is bipartite, and that all the vertices of $FSM \SMP G$ are \emph{half-integral} in general, meaning that $x_e \in \left\{0,\frac{1}{2},1\right\}$ for each $e \in E$ for each vertex $x$. 
	
Irving \cite{irving1985} gave a two-phase polynomial-time algorithm for finding a stable matching of an instance $\SMP G$ or determining that no such matching exists. Phase one of this algorithm consists of a proposal sequence that results in a unique subgraph of $G$ that has the same set of stable matchings as $\SMP G$.
	
\begin{definition}[Irving's phase one]	For an instance $\SMP G$ of SMP let $G_I$ be the subgraph of $G$ obtained by applying phase one of Irving's algorithm to $\SMP G$. 
\end{definition}
Chen et al. \cite{chen2012} show that $FSM \SMP G$ is integral if and only if the graph $G_I$ is bipartite.  Note that, for bipartite graphs $G$, the integrality of $FSM \SMP G$ follows easily from this result.
	\begin{theorem}[Chen \cite{chen2012}]
		The following are equivalent.
		\begin{enumerate}
			\item $FSM \SMP G$ is integral.
			\item $G_I$ is a bipartite graph.
		\end{enumerate}
\label{thm:chen}
\end{theorem}
	
We will use the following well known property of stable matchings, which says that whenever an edge is in the support of some fractional stable matching then the corresponding stability constraint for this edge is satisfied with equality. 
	
	\begin{lemma} [Abeledo, Rothblum \cite{abeledo1994}]
		Let $x \in FSM \SMP G$. Then for any $e \in E$, 
		\begin{equation}
		x_e > 0 \;\Rightarrow \;x(\phi(e)) = 1.
		\end{equation}
		\label{lemma:AR_xe_positive}
	\end{lemma}
In the same paper, Albedo and Rothblum show that the set of matched vertices is the same for every stable matching $M$ of $\SMP G$. 
\begin{lemma}[Abeledo and Rothblum \cite{abeledo1994}]
		Given $\SMP G$, the vertex set $V$ can be partitioned into two sets $V^0$ and $V^1$ such that
		\begin{eqnarray*}
			&& V^0 = \{v \in V: v \text{ is not matched in any stable matching } M \text{ of } \SMP G\} \\
			&& V^1 = \{v \in V: v \text{ is matched in every stable matching } M \text{ of } \SMP G\}. 
		\end{eqnarray*}
		Moreover, for each $j = 0,1$ and $v \in V^j$, it holds that $x(\delta(v)) = j, \;\; \forall \;x \in FSM \SMP G$.
		\label{lemma:AR_V0_V1}
\end{lemma}
	
For our purposes, it will be easier to work with instance where stable matchings are perfect, meaning that $V_0 = \emptyset$. We show that we can focus on such instances without any loss of generality. First, we can use Irving's algorithm to decide if stable matchings exists, and if so, find the bipartition of the agents into $V_0$ and $V_1$ as in Lemma \ref{lemma:AR_V0_V1}.
Now, consider the subgraph $G[V_1]$ of $G$. Clearly, any stable matching of $G$ is a stable matching of $G[V_1]$. Moreover, a stable matching of $G[V_1]$ will be a stable matching of $G$ as long as none of the edges with an endpoint in $V_0$ are blocking, meaning that any vertex $u$ which has a neighbor $w \in V_0$ is matched to someone they prefer over $w$. We thus obtain the following observation.

\begin{theorem}[Reducing to the perfect stable matching case]
		Given $\SMP G$, let $V^0$ and $V^1$ be as in Lemma \ref{lemma:AR_V0_V1}. 
		Let
		$$E^1 = \left\{ uv \in E: \begin{aligned} &u,v\in V^1 \\ &\forall \, w \in V^0, \;uw \not\in E \text{ or } v \pref u w \\ &\forall \, w \in V^0, \;vw \not\in E \text{ or } u \pref v w \end{aligned} \right\}$$
		and  $G^1 = (V^1,E^1)$. Then $M$ is a stable matching of $\SMP G$ if and only if it is a perfect stable matching of $\SMP {G^1}$.
		\label{thm:reduce_to_perfect_matchings}
\end{theorem}
Observe that, by Lemma \ref{lemma:AR_V0_V1}, if an instance $\SMP G$ has a perfect stable matching, then all stable matchings of $\SMP G$ are perfect. Thus, from now on, we consider without loss of generality only instances which have at least one stable matching, and all the stable matchings are perfect.

In our analysis we will rely on a useful characterization of the vertices of $FSM \SMP G$  introduced by Chen et al. \cite{chen2012}. This characterization is based on the notion of semi-stable partitions, which we define below. A cycle $C=v_0v_1 \cdots v_{k-1} v_0$ of $G$ is said to have a \emph{cyclic preference} if $v_{i-1} \prec_{v_i} v_{i+1}$ for all $0 \leq i \leq k-1$ or $v_{i+1} \prec_{v_i} v_{i-1}$ for all $ 0 \leq i \leq k-1$, where all indices are taken modulo $k$. 
\begin{definition}[Semi-stable partitions.]\label{def:semi} Given $\SMP G$ let $\mathcal C = \{C_1,\dots,C_k\}$ be a set where
	\begin{itemize}
		\item each $C_i$ is an edge of $G$ or a cycle in $G$ with cyclic preferences in $\SMP G$,
		\item each vertex $v \in V$ is contained in precisely one $C_i$,
	\end{itemize}
	and let $x_{\mathcal C} \in \R^{E}$ be defined as
	$$(x_{\mathcal C})_e = 
	\begin{cases} 1 & \text{ if } e \text{ is an edge in } \mathcal C \\ 
	\frac{1}{2} & \text{ if } e \text{ is a part of a cycle } C \text{ in } \mathcal C \\
	0 & \text{ otherwise } \end{cases} $$
Then  $\mathcal C$ is a \emph{semi-stable partition} if $x_\mathcal C \in \SMP G$. 
\end{definition}
In other words, for $\mathcal C$ to be a semi-stable partition, $\mathcal C$ must be a collection of edges and cycles with cyclic preferences, such that the point $x$ assigning value $1$ to the single edges and $\frac{1}{2}$ to the edges on the cycles is a feasible point of $FSM \SMP G$. Chen et al. \cite{chen2012} show that if $x$ is a vertex of $FSM \SMP G$, then $x = x_\mathcal C$ for some semi-stable partition $\mathcal C$ .
	
	\section{Identifying redundant edges}
	\label{section:identifying_redudant_edges}
	
	Chen et al. \cite{chen2012} prove that the instances $\SMP G$, for which the polytope $FSM \SMP G$ is integral, are precisely those, for which the graph $G_I$ obtained from phase one of Irving's algorithm is bipartite. Thus, the polytope $FSM \SMP G$ can be used to solve the weighted stable matching problem only for this class of instances. The edges deleted in phase one are, in a sense, redundant for the instance. By this we mean that none of these edges appear in any stable matching of $G$, and also by deleting them, the set of stable matchings of the instance stays unchanged, i.e. $\SMP G$ and $\SMP {G_I}$ have the same set of stable matchings. However, it could still be the case that $G_I$ contains edges that never appear in any stable matching of $\SMP G$. This suggests the one could improve upon phase one of Irving's algorithm by further searching for edges that never appear in a stable matching and whose removal preserves the set of stable matchings. 
	
	In general, when removing an edge $e$ from a graph $G$ that is never a part of a stable matching, the set of stable matchings either stays the same, or gets larger. All the original stable matchings of $\SMP G$ are still stable matchings of $(G \setminus e,\prec)$; however, it can happen that a matching $M'$ is stable in $(G \setminus e,\prec)$, but is not stable in $\SMP G$, because $e$ is a blocking edge for $M'$ in $\SMP G$. Thus our goal is to identify edges $e$ that are not in any stable matching, and whose removal does not introduce any extra stable matchings. As a motivation for this approach, consider the following example. 
	
	\begin{example}
		\label{ex:diff_between_phase1_and_alg}
		Let the preferences be as follows, where each line $v|v_1 \dots v_k$ means $v_1 \pref v \dots  \pref v v_k$
			\begin{center}
			\begin{tabular}{c|ccccc}
				1 & 3 & 4 & 5 & 2 \\
				2 & 1 & 4 & 3 & 5 & 6 \\
				3 & 5 & 6 & 1 & 2 \\
				4 & 5 & 6 & 1 & 2 \\
				5 & 1 & 2 & 3 & 4 \\
				6 & 2 & 3 & 4 
			\end{tabular}
		    \end{center}
During phase one, only edges 12, 23 and 45 get eliminated. Figure \ref{fig1} depicts the original graph $G$ on the left and the graph $G_I$ in the middle. The red numbers represent the ranks of the edges by their respective endpoints.
\begin{figure}[ht!]
\begin{center}
	\subfloat{\includegraphics[scale=0.33]{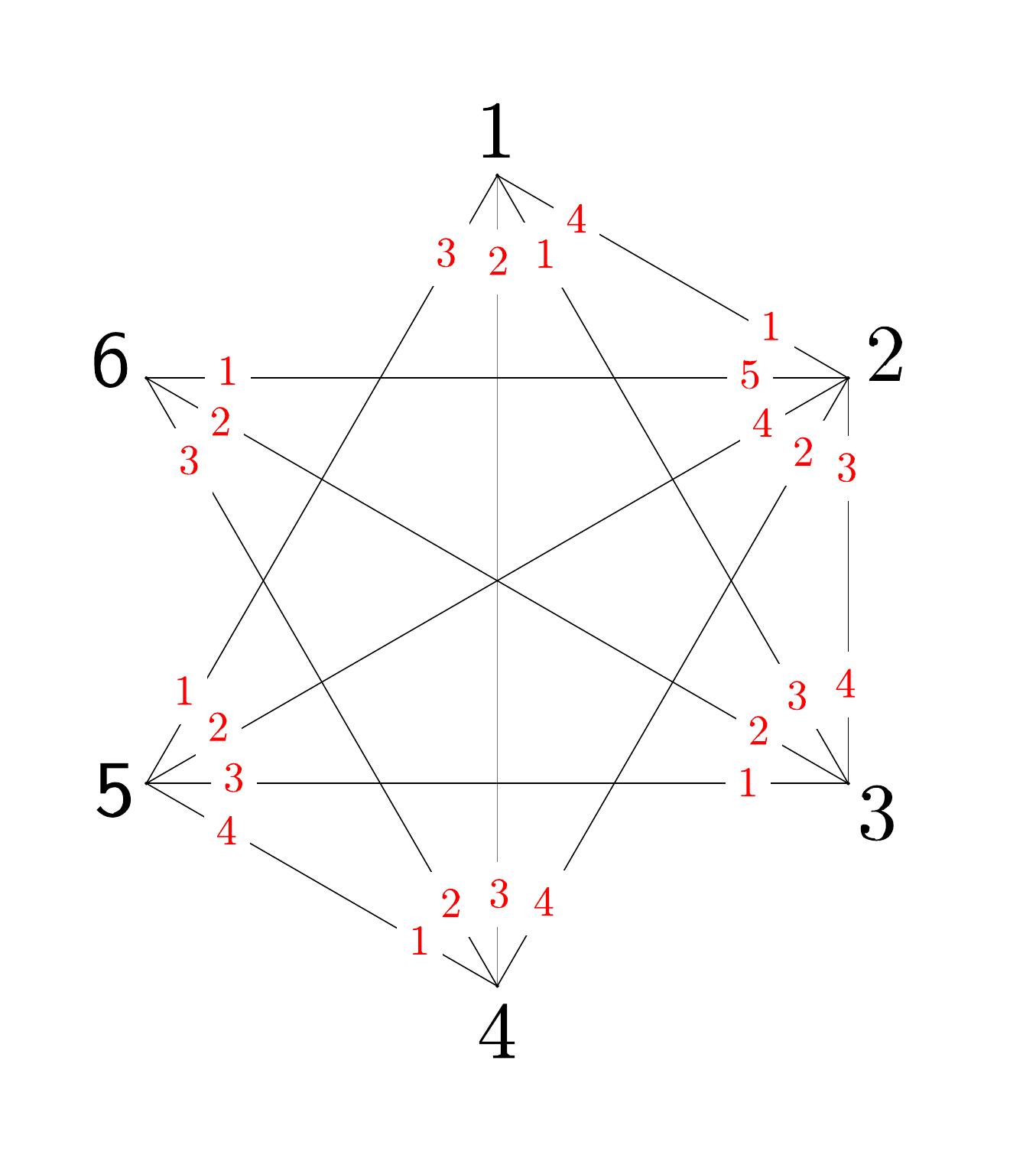}}
	\hspace{10pt} 
	\subfloat{\includegraphics[scale=0.33]{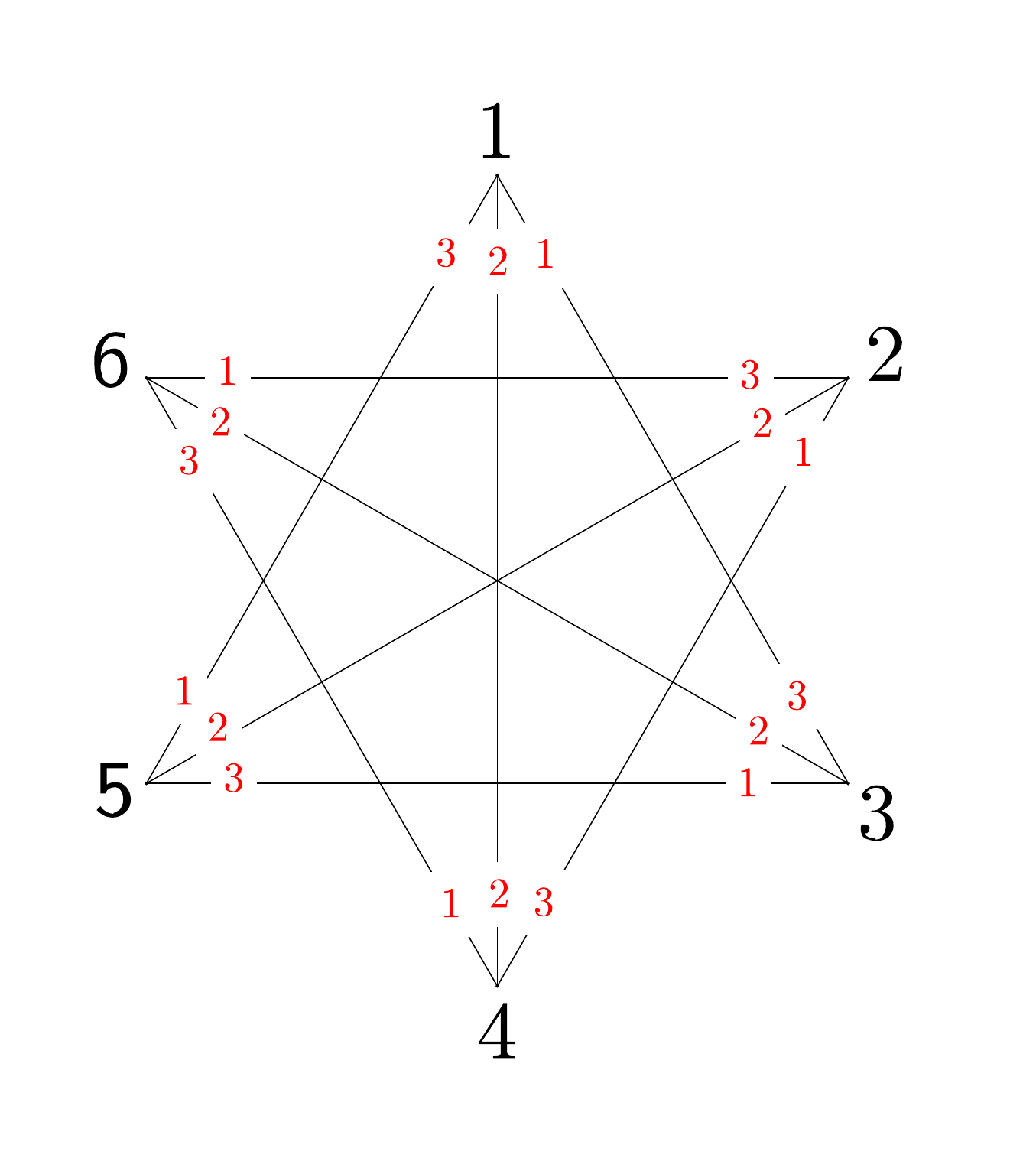}}
	\hspace{10pt} 
	\subfloat{\includegraphics[scale=0.33]{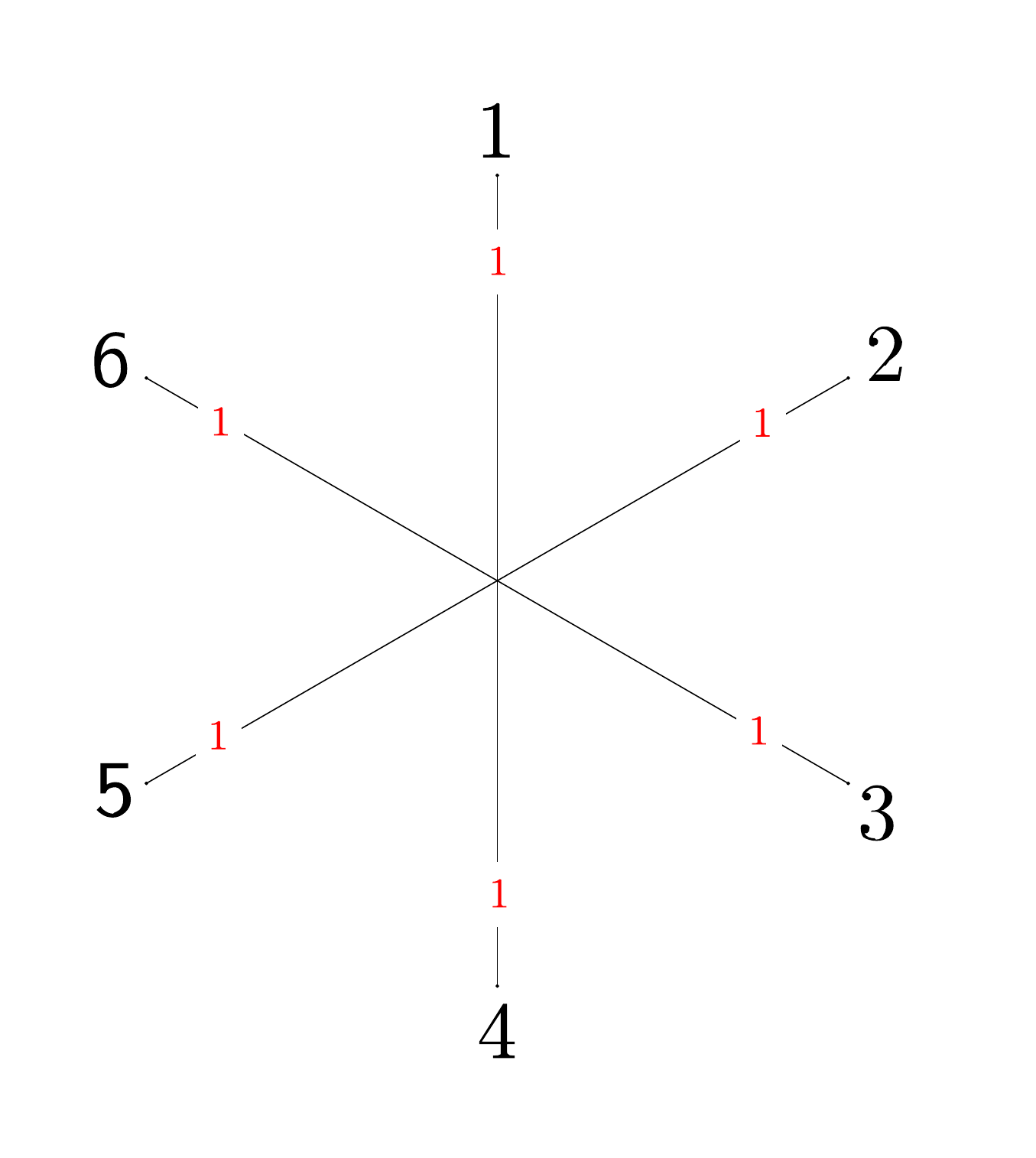}}
\end{center}
\caption{The graphs $G$, $G_I$ and $H$ from Example \ref{ex:diff_between_phase1_and_alg}. } 
\end{figure}\label{fig1}
It is easy to check that $FSM \SMP {G} = conv\{x,y\}$, where 
\begin{align*}
x_e &= \begin{cases} 1 & e \in \{ 14, \, 25, \, 36 \} \\ 0 & \text{ otherwise } \end{cases}  \quad \quad \text{and} \quad \quad 
y_e = \begin{cases} \frac{1}{2} & e \in \{ 13, \, 35, \, 15, \, 24, \, 46, \, 26 \} \\ 0 & \text{ otherwise } \end{cases} 
\end{align*}
The edge 12 is never a part of any stable matching, and, at the same time, the worst (remaining) option for agent 1. By removing this edge we don't introduce any new stable matchings, since any stable matching $M$ in $\SMP {G \setminus \{12\}}$ is perfect, and $1$ has a strictly better partner in $M$ than 2, thus $M$ is stable in $\SMP G$, too. By repeating the above argument, we can keep reducing our graph by eliminating edges 12, 26, 32, 42, 54, 64, 15, 31 and 53 in turn, while always preserving the same set of stable matchings. The resulting graph $H$ consists of precisely the three edges 14, 25, 36 that form a stable matching - the only stable matching of the original instance $\SMP G$. 

\end{example}
This example illustrates that there are cases where even though $G_I$ is not bipartite and $FSM \SMP G$ is not integral, one can still identify redundant edges in order to arrive at a bipartite subgraph, whose set of stable matchings is the same as the original instance meaning that $wSMP$ can be solved in polynomial time. Our goal is to develop a polynomial-time algorithm that finds such a subgraph whenever it exists.

We start by showing that it is possible to identify the set of edges that appear in \emph{a} stable matching of $\SMP G$ in polynomial time using Irving's algorithm.
	
	\begin{definition}
		Given $\SMP G$, let $$E_M = \{e \in E: \exists \text{ stable matching } M \text{ s.t. } e \in M\}$$ be the set of edges of $G$ that ever appear in a stable matching of $G$. 
		Further, let $\overline H = (V, E_M)$ be the graph consisting only of the edges in $E_M$.
		\label{def:EM}
	\end{definition}

The graph $\overline{H}$ is a subgraph of the graph $G_I$, and it shares the following useful property with $G_I$. 
			
	\begin{lemma}
		Let $\SMP G$ be given and  $uv \in E_M$. Then $u = f_{\overline H}(v) \Leftrightarrow v = l_{\overline H} (u)$.
		\label{lemma:best_worst_edges_in_EM}
	\end{lemma}
We remark however, that while $(G_I, \prec)$ keeps the same stable matchings as $(G, \prec)$, the set of stable matchings of $(\overline{H}, \prec)$ can be a strict superset of the set of stable matchings of $(G, \prec)$. 

The following lemma shows that we can decide in polynomial time whether a given edge $e \in E$ belongs to the set $E_M$. 

%
	
	\begin{lemma}
		Given an instance $\SMP G$ and an edge $uv \in G$, there exists a stable matching containing $uv$
		 if, and only if, $\SMP {G' = (V', E')}$ contains a stable matching, where 
		\begin{eqnarray*}
			V' &=& V \setminus \{u,v\} \\
			E_u &=& \{wy \in G[V']: uw \in E, w\pref u v, u \pref w y \}  \\
			E_v &=& \{wy \in G[V']: vw \in E, w \pref v u, v \pref w y \} \\
			E' &=& E(G[V']) \setminus E_u \setminus E_v.
		\end{eqnarray*}
		\label{lemma:edges_in_a_SM}
	\end{lemma}
In particular, this implies that the set of edges that are part of a stable matching can be identified in polynomial time. 	\footnote{ A more general version of the problem of deciding whether a given set of edges is contained in a stable matching was considered by Dias et al \cite{dias2003}.}

As already mentioned, by removing an edge $e \in E \setminus E_M$ from a graph $G$, we can introduce new stable matchings. The following lemma tells us that we can check whether this is the case or not. 
	
	\begin{lemma}
		Given $\SMP G$, let $E_M$ be defined as above and let $e \in E \setminus E_M$. 
		Then it can be checked in polynomial time whether the set of stable matchings of $\SMP G$ and the set of stable matchings of $\SMP {G \setminus e}$ are the same.
		\label{lemma:chceck_removing_edge_no_effect}
	\end{lemma}
To see why Lemma \ref{lemma:chceck_removing_edge_no_effect} is true, 
	observe that for a matching $M'$ to be stable in $\SMP {G \setminus \{uv\}}$, but not in $\SMP G$,both $u$ and $v$ must prefer each other over their respective partners in $M'$. For each such an assignment of partners for $u$ and $v$, we can check whether it can be extended to a perfect stable matching in $\SMP {G \setminus \{uv\}}$ as in \cite{dias2003}. 
	
	Note in particular, that if an edge $e$ is the worst choice for one of its endpoints, then it can never be blocking for any perfect matching. If in addition $e \in E \setminus E_M$ then removing it keeps the set of stable matchings the same. We thus have the following lemma, whose proof is in the Appendix.
	
	\begin{lemma}
		Let $\SMP G$ be given and let an edge $e = uv \in E \setminus E_M$ be such that $u = l_G(v)$. 
		Then $\SMP G$ and $(G \setminus e,\prec)$ have the same set of stable matchings.
		\label{lemma:worst_edges_redundant}
	\end{lemma}
The proofs of all the lemmas in this section are deferred to the Appendix. 
	
	\subsection{Algorithm}
	\label{section:removing_redundant_edges}
	
	Instead of applying phase one of Irving's algorithm to a given instance $\SMP G$, we will now use Lemma \ref{lemma:worst_edges_redundant}
	to keep identifying and removing edges from the graph while preserving the set of stable matchings. 
	
	\begin{algorithm}\textbf{Algorithm 1.} 
		Given $\SMP G$, 
		\begin{enumerate}
			\item[1.] Identify the set $E_M$ of edges contained in stable matchings of $\SMP G$.
			\item[2.] While possible:
			\begin{enumerate}
				\item[(i)] Identify an edge $e \in E(G)$ such that
				\begin{itemize}
					\item[$\circ$] $e \not \in E_M$, and
					\item[$\circ$] $e$ is the worst remaining option for one of its endpoints.
				\end{itemize}
				\item[(ii)] Remove $e$ from $G$: set $G \leftarrow G \setminus e$.
			\end{enumerate}
			\item[3.] Return $G$.
		\end{enumerate}
	\end{algorithm}
	
	\begin{definition}
		Given $\SMP G$, let $H$ be the graph obtained from $G$ by applying the Algorithm 1 to it. Denote $E_H$ the set of its edges.
		\label{def:H}
	\end{definition}
First thing to observe is that the graph we obtain from the Algorithm 1 is independent of the order of deletion of the edges. In other words, $H$ is well-defined.
	
	\begin{lemma}
		The graph $H$ is well-defined.
		\label{lemma:H_well_defined}
	\end{lemma}
The main idea behind the proof of Lemma \ref{lemma:H_well_defined} is that once an edge $e$ can be removed from $G$ in Algorithm 1, it can also be removed from any graph $G'$ otained from $G$ by a sequence of edge removals as in Algorithm 1. Thus $e$ will be removed from $G$ no matter what the order of deletion of the edges is. The full proof of Lemma \ref{lemma:H_well_defined} can be found in the Appendix. 

\newpage
We now enumerate several properties of $H$. The proof is found in the Appendix.
	\begin{theorem}
		Given $\SMP G$, let $E_M$ and $H$. Then
		\begin{enumerate}[(i)]
			\item $E_M \subseteq H$.
			\item $\SMP G$ and $\SMP H$ have the same set of stable matchings. 
			\item $H$ is a subgraph of $G_I$.
			\item $H = H_I$, where $H_I$ is obtained from $H$ by applying phase one to $\SMP H$.
			\item Let $uv \in E_H$. Then $u = f_H(v) \Leftrightarrow v = l_H(u)$. 
			\item Let $v \in V$. Then $vf_H(v) \in E_M$.
			\item Let $uv \in E_H \setminus E_M$. Then $f_H(v) \pref v u \pref v l_H(v)$ and $f_H(u) \pref u v \pref u l_H(u)$.
		\end{enumerate}
		\label{thm:properties_of_H}
	\end{theorem}
	

	Note in particular that, for each agent, his worst and best remaining options are in $E_M$. That means that, in the algorithm, we were deleting edges that were worse than edges in $E_M$ for one of the endpoints, but by doing this, all the edges that were better for one of the endpoints than edges in $E_M$ got removed as well. 
	
	
Since $\SMP G$ and $\SMP H$ have the same set of stable matchings by construction of $H$, the polytopes $FSM \SMP G$ and $FSM \SMP H$ have exactly the same set of integral vertices. Moreover, if $H$ is bipartite, then Theorem \ref{thm:chen} implies that $FSM \SMP H$ is integral. In general, it can happen that $FSM \SMP G$ is not integral, but $FSM \SMP H$ is, which happens precisely when $G_I$ is not bipartite, but $H$ is. If this is the case, then there exists at least one semi-stable partition $\mathcal C$ such that $x_\mathcal C$ is a fractional vertex of $FSM \SMP G$, and $supp(x) \nsubseteq E_H$. Moreover, no semi-stable partition $\mathcal C$ corresponding to fractional vertex of $FSM \SMP G$ satisfies $E(\mathcal C) \subseteq E_H$.	
	
	\section{Characterizing bipartite reducible graphs}
	\label{section:GI_vs_H}

In the Algorithm 1, we removed \emph{some} edges such that their deletion didn't affect the set of stable matchings. However, it is not guaranteed that this deletion couldn't continue further, or that some other method of eliminating edges from $G$ to obtain a smaller subgraph of $G$ with the same set of stable matching wouldn't yield a graph $K$ whose polytope $FSM \SMP K$ is integral. 

In this section, we will show that if some order of deletion would yield such a $K$, or, in fact, if any other method of obtaining such a $K$ worked, then $H$ would be bipartite and the corresponding polytope would be integral.  Our goal is to prove the following theorem.
	
	\begin{theorem}
		Let $\SMP G$ be given. Then if $H$ is not bipartite, then neither is any subgraph $K$ of $G$ with the property that $\SMP {K}$ and $\SMP G$ have the same set of stable matchings. 
		In other words, $H$ is bipartite if and only if $\SMP G$ is bipartite reducible. 
		\label{thm:existence_of_bipartite_representation}
	\end{theorem}
	
	\subsection{Modifying the FSM polytope}
	\label{section:modigying_FSM}
	
	We start by having a closer look at the fractional stable matching polytope $FSM \SMP G$. It is easy to see that each edge $e \in E \setminus E(G_I)$ satisfies $x_e = 0$ for all points $x \in FSM \SMP G$, but it is in general not true that $e \not\in E_M \Rightarrow x_e = 0,\; \forall x \in FSM \SMP G$. Thus we start by adding the constraint  $x_e = 0$ for each edge $e \in E \setminus E_M$ to $FSM \SMP G$.
	
	\begin{definition}
		Given a system $\SMP G$ with $E_M$ as above, let 
		$${FSM}' \SMP G = \left\{ x \in \R^{E} :\;\; \begin{array}{rl} 
		x_e = 0 & \forall e \in E \setminus E_M  \\ 
		x(\delta(v) \cap E_M) \leq 1 & \forall v \in V \\ 
		x (\phi(e) \cap E_M) \geq 1 & \forall e \in E \\ 
		x \geq 0 &  \end{array}\right\}.$$
		$$\overline{FSM} \SMP G = \left\{ x \in \R^{E_M} : \begin{array}{rl} 
		x(\delta(v) \cap E_M) \leq 1 & \forall v \in V \\ 
		x (\phi(e) \cap E_M) \geq 1 & \forall e \in E \;\;\;\;\;\;\; \\ 
		x \geq 0 &  \end{array}\right\}.$$ 
		\label{def:overlineFSM}
	\end{definition}
	
	We observe that $FSM' \SMP G$ is the intersection of $FSM \SMP G$ with some of its facets - namely the facets $\{x_e = 0\}_{e \in E \setminus E_M}$. Further,
	 $\overline{FSM} \SMP G$ is just a restriction of $FSM' \SMP G$ to the edges in $E_M$,
     where we drop coordinates in $E \setminus E_M$, which satisfy $x_e = 0 \;\forall x \in FSM' \SMP G$. 	Thus, the following holds:
	
	\begin{remark}
		For an arbitrary instance $\SMP G$,
		\begin{itemize} 
			\item each vertex of $FSM'\SMP G$ is also a vertex of $FSM \SMP G$,
			\item all vertices of $FSM'\SMP G$ are half-integral,
			\item the integral vertices of $FSM'\SMP G$ are precisely the integral vertices of $FSM \SMP G$,
			\item there is a 1-to-1 correspondence between the vertices of $\overline{FSM}\SMP G$ and those of $FSM' \SMP{G}$. 
				Each vertex of $\overline{FSM}\SMP G$ is a vertex of $FSM' \SMP G$ projected to $E_M$. 
		\end{itemize}
		\label{remark}
	\end{remark}
Following these remarks, we can observe that the integral vertices of $\overline{FSM} \SMP G$ correspond to stable matchings of $\SMP G$. Moreover, it can be the case that by intersecting the polytope $FSM \SMP G$ with its facets $\{x_e = 0\}$ for all $e \in E \setminus E_M$ we cut off all its fractional vertices, and the polytope $\overline{FSM} \SMP G$ is integral. 
	
	
	Let us revisit Lemma \ref{lemma:worst_edges_redundant} from the previous section in the light of the polytope $\overline{FSM} \SMP G$. We defer the proof to the Appendix.
	
	\begin{lemma}
		Let $\SMP G$ be given and let an edge $e = uv \in E \setminus E_M$ be such that $u = l_G(v)$. Then 
		\begin{enumerate}[(i)]
			\item $\SMP G$ and $(G \setminus e,\prec)$ have the same set of stable matchings.
			\item $\overline{FSM} \SMP G = \overline{FSM} \SMP {G \setminus e}$. 
			In other words, the inequality $x(\phi(e) \cap E_M) \geq 1$ is redundant in $\overline{FSM} \SMP G$.
		\end{enumerate} 
		\label{lemma:worst_edges_redundant_inequalities}
	\end{lemma}
Note that the removal of an edge that is the worst remaining choice for one of its endpoints does not change the polytope $\overline{FSM} \SMP G$. This leads us to the following result.
	
	\begin{corollary}
		For any $\SMP G$, we have that $\overline{FSM} \SMP G = \overline{FSM} \SMP H$.
		\label{corr:ovelineG_overlineH}	
	\end{corollary}
	
	\subsection{Conditions for integrality of $\overline{FSM} \SMP G$}
	\label{section:integrality_of_overlineFSM}
	
	We now give sufficient and necessary conditions for integrality of $\overline{FSM} \SMP G$. For the rest of this section, we will consider semi-stable partitions with respect to the polytope $\overline{FSM} \SMP G$ (note that in Definition \ref{def:semi} they were defined with respect to $FSM \SMP G$).  Analogously to our original definition, we will say that $\mathcal C$ is a semi-stable partition w.r.t $\overline{FSM} \SMP G$, if and only if $\mathcal C$ is a collection of vertex-disjoint edges and cycles with cyclic preferences in $E_M$, such that the point $x_{\mathcal C}$ assigning value 1 to the single edges and $\frac{1}{2}$ to the edges on the cycles is a feasible point of $\overline{FSM} \SMP G$.
	
	We let $E(\mathcal C)$ be the edges in $\mathcal C$, and we use the notation of Chen et al. \cite{chen2012} to define the set of edges $$E_{\mathcal C} = \{ uv \in E_H: \exists \, e_1,e_2,f_1,f_2 \in E(\mathcal C): e_1 \preceq_u uv \preceq_u e_2 \text{ and } f_1 \preceq_v uv \preceq_v f_2  \}.$$
	 Observe that $E_{\mathcal C}$ consists of all edges in $\mathcal C$ and of all so-called intermediate edges, which are edges with both endpoints on a cycles from $\mathcal C$ that fit in between the edges in $\mathcal C$ in both endpoint's preference lists. We then let $H_{\mathcal C}$ be the subgraph of $H$ induced by the edges $E_{\mathcal C}$. 
	 
	 The following lemma, whose proof is in the Appendix, characterizes the vertices of $\overline{FSM} \SMP G$.
	 
\begin{lemma}
		Let $x$ be a vertex of $\overline{FSM} \SMP G$. Then $x = x_\mathcal C$ for some semi-stable partition $\mathcal C$ w.r.t. $\overline{FSM} \SMP G$.
		\label{lemma:vertices_semistable_partitions}
	\end{lemma}
With this, we obtain the following equivalent conditions for the integrality of $\overline{FSM} \SMP G$.
\begin{theorem}
		The following are equivalent.
		\begin{enumerate}
			\item $\overline{FSM} \SMP G$ is integral,
			\item $H$ is bipartite,
			\item $H_\mathcal C$ is bipartite for any semi-stable partition $\mathcal C$ w.r.t. $\overline{FSM} \SMP G$.
		\end{enumerate}
		Moreover, if $x$ is a vertex of $\overline{FSM} \SMP G$, then $x$ is half-integral.
		\label{thm:integrality_of_overlineFSM}
	\end{theorem}
We note that both Lemma \ref{lemma:vertices_semistable_partitions} and Theorem \ref{thm:integrality_of_overlineFSM} are analogous to the results proven by Chen et al. \cite{chen2012} for $FSM \SMP G$.
	
	Note that, in particular, if $H_\mathcal {C}$ is not bipartite for some semi-stable partition $\mathcal{C}$ w.r.t $\overline{FSM} \SMP G$, then $\overline{FSM} \SMP G$ is not integral.
	On the other hand, if $\overline{FSM} \SMP G$ is integral, then all the graphs $H_{\mathcal C}$ have the same bipartition, namely the bipartition of $H$. To make the proof more readable, we will refer to semi-stable partitions w.r.t. $\overline{FSM} \SMP G$, simply as semi-stable partitions. 
	
	\smallskip
	\textit{Proof of Theorem \ref{thm:integrality_of_overlineFSM}.}
	
	$(2) \Rightarrow (1)$. 
		If $H$ is bipartite, then $FSM \SMP {H}$ is integral, and so is $\overline{FSM} \SMP H$, by Remark \ref{remark}. 
		By Corollary \ref{corr:ovelineG_overlineH}, $\overline{FSM} \SMP G = \overline{FSM} \SMP H$.
		
	$(3) \Rightarrow (2)$.
		To see this, note that the set $F = \{vf_H(v): v \in V\}$ of all the edges most preferred in $H$ by one of its endpoints is a semi-stable partition. 
		Moreover, $H_F = H$. 
		
    $(1) \Rightarrow (3)$.
		Assume $\overline{FSM} \SMP G$ is integral. By Lemma \ref{lemma:vertices_semistable_partitions}, $\overline{FSM} \SMP G$ is integral if and only if 
		for each semi-stable partition $\mathcal C$, $x_\mathcal C$ is a convex combination of integral points in $\overline{FSM} \SMP G$.

		Fix a semi-stable partition $\mathcal C$. Then we can write $x_\mathcal C = \sum_{i = 1}^k \lambda_i y^i$, where $k \geq 1$, $\lambda_i > 0$, $\sum_{i = 1}^k \lambda_i = 1$ 
		and $y^i$'s are all integral points in $\overline{FSM} \SMP G$. For each $i \in [k]$, denote by $M^i$ the perfect stable matching consisting of edges in $supp(y^i)$.
		For simplicity, we will denote $x_\mathcal C$ by $x$. 
		
		If $k = 1$, then $x = y^1$, thus $\mathcal C$ is a matching. Then $E(H_{\mathcal C}) = \mathcal C$, and $H_\mathcal C$ is bipartite. 
		Hence we can assume $k \geq 2$.	Let $\{e_i = u_iv_i\}_{i \in [m]}$ be the set of single edges in $\mathcal C$, and let 
		$W = V \setminus \{u_i, v_i\}_{i \in [m]}$ be the vertices on cycles in $\mathcal C$.
		
		For each $i \in [k]$ we have 
		$x_e = 0 \Rightarrow y^i_e = 0 \text{    and    } x_e = 1 \Rightarrow y^i_e = 1$,
		so in particular $supp(y^i) \subseteq supp(x)$. 
		Since for each $i \in [k]$, $M^i$ must be a perfect matching whose edges are a subset of $E(\mathcal C)$, 
		we immediately have that all the cycles in $\mathcal C$ must be even. 
		
		Further note that each $v \in W$ is matched to one of only two possible $u$'s in any of the $M^i$'s, namely to one of its two neighbors on the cycle in $\mathcal C$ it is a part of. 
		Let denote these two partners $v_+$ and $v_-$, where $v_+ \pref v v_-$. 
		Then let $W_1 = \{v \in W: {vv_+} \in M^1\}$ be the set of agents that are matched to their more preferred neighbor in $M^1$.
		Let $W_2 = W \setminus W_1$.
		
		\textit{Claim: $(W_1 \cup \{u_i\},W_2 \cup \{v_i\})$ is a bipartition of $H_\mathcal C$.}
		
		\textit{Proof:} Note that $u = v_+ \Leftrightarrow v = u_-$ and, in particular, the only edges of $H_\mathcal C$ that could lie in $W_1 \times W_1$ or in $W_2 \times W_2$ 
		are those in $E(H_{\mathcal C}) \setminus E(\mathcal C)$. 
		Further, for each $uv$ in $E(H_{\mathcal C}) \setminus E(\mathcal C)$, $(\phi(uv) \cap supp(x)) = \{uu_+, vv_+\}$.
		
		If there was an edge $uv \in W_2 \times W_2$, then $y^1(\phi(uv)) = 0$ since neither $uu_+$ nor $vv_+$ gets value 1 in $y^1$, 
		so the stability condition at $uv$ would be violated.
		
		Finally, suppose there was an edge $uv \in W_1 \times W_1$. Let $z = \sum_{i = 2}^k \frac{\lambda_i}{(1 - \lambda_1)} y^i$. 
		Then $z \in \overline{FSM} \SMP G$ and $x = \lambda_1 y^1 + (1 - \lambda_1)z$. Thus $z = \frac{1}{1 - \lambda_1}(x - \lambda_1 y)$ and 
		$$z(\phi(uv)) = \frac{1}{1 - \lambda_1} \left( (x_{uu_+} - \lambda_1 y_{uu_+}) + (x_{vv_+} - \lambda_1 y_{vv_+}) \right) = 
		\frac{2(\frac{1}{2} - \lambda_1) }{1 - \lambda_1} = \frac{1 - 2 \lambda_1}{1 - \lambda_1} < 1.$$
		This finishes the proof of the claim and of the theorem. \qed
	
	\subsection{Proof of Theorem \ref{thm:existence_of_bipartite_representation}}
	\label{section:pf_of_bipartiet_representation}
	
	We are now ready to prove Theorem \ref{thm:existence_of_bipartite_representation}.
	 
	\smallskip
	\textit{Proof of Theorem \ref{thm:existence_of_bipartite_representation}. }
		Suppose that $K$ is a bipartite subgraph of $G$ such that $\SMP G$ and $\SMP K$ have the same set of stable matchings. 
		Let $K'$ be the graph obtained from $K$ by applying Algorithm 1 to it. Since $K'$ is bipartite, 
		by Theorem \ref{thm:integrality_of_overlineFSM} applied to $K$, $\overline{FSM} \SMP K$ is integral, and so is $\overline{FSM} \SMP {K'}$ by 
		Corollary \ref{corr:ovelineG_overlineH}.
		
		Since in $K'$, for each agent we have that its best and worst remaining edge must be in $E_M$, 
		$K'$ contains no edges in $E \setminus E_H$, thus $E(K') \subseteq E_H$. 
		But then all the inequalities from $\overline{FSM} \SMP {K'}$ are in $\overline{FSM} \SMP H$, 
		meaning that $\overline{FSM} \SMP {K'} \supseteq \overline{FSM} \SMP H$. 
		Further, $\SMP H$ and $\SMP {K'}$ have the same set of stable matchings, so the integral vertices of $\overline{FSM} \SMP H$ are the integral vertices of $\overline{FSM} \SMP {K'}$ and vice versa. 
		Hence $\overline{FSM} \SMP H$ is integral.
		 
		By Corollary \ref{corr:ovelineG_overlineH} and Theorem \ref{thm:integrality_of_overlineFSM}, $H$ is bipartite. \qed
	
	\section{Approximation algorithm}
	\label{section:EM_degree_two}
	
	In this section, we study instances where each agent has at most two possible partners in any stable matching, meaning that $\overline H$ consists only of single edges and disjoint cycles with cyclic preferences. Note that even in these instances $\overline H$ being bipartite is not enough to ensure that $\overline {FSM} \SMP G$ is integral. However, we show that we can obtain a 2-approximation algorithm for the minimum-weight stable matching in this case.
	
	\begin{theorem}
		Suppose $\SMP G$ is such that $\overline H$ consists only of single edges and vertex-disjoint cycles. 
		Then, for any weight function $w: E \rightarrow \R_{\geq 0}$, we can find, in polynomial time, a stable matching $M^\star$ such that $w(M^\star) \leq 2 w(M_{OPT})$, where $M_{OPT}$ is a stable matching of $\SMP G$ of minimum weight.
		\label{thm:two_approx}
	\end{theorem}
	
	\begin{proof}
		Suppose that $\overline H$ consists of vertex-disjoint cycles $C_1,\dots, C_k$ and single edges $e_1,\dots, e_m$. 
		Then, by definition of $\overline H$, these cycles must all have even length 
		(otherwise it would not be true that each edge in $\overline H$ is a part of a stable matching) 
		and also that the preferences on the cycles are cyclic (following from Theorem \ref{thm:properties_of_H}).
		
		Let $M$ be a stable matching of $\SMP G$. Then $e_i \in M \;\forall i \in [m]$ and also $|C_k \cap M| = \frac{1}{2}|C_k|, \;\forall i \in [k]$. 
		Hence, for any cycle $C_i = v_i^1,\dots, v_i^{2r_i}$ we have either that $\{v_i^{2j}v_i^{2j+1}\}_{j \in [r_i]} \subseteq M$ or $\{v_i^{2j-1}v_i^{2j}\}_{j \in [r_i]} \subseteq M$. 
		For each $i \in [k]$ let $$w_i^+ = \sum_{j \in [r_i]} w(v_i^{2j-1}v_i^{2j}) \;\;\;\;\;\;\;\;\;\;\text{and}\;\;\;\;\;\;\;\;\;\;w_i^- = \sum_{j \in [r_i]} w(v_i^{2j}v_i^{2j+1})$$ 
		and suppose without loss of generality that $w_i^+ \geq w_i^-$. 
		
		We construct a weight function $\tilde{w}$ as follows: 
		For each $i \in [k]$ we set $\tilde{w} (v_i^1v_i^2) = w_i^+ - w_i^-$ and we set the weights of all the other edges to zero. 
		Intuitively, for each stable matching $M$ and for each cycle $C_i$ we must pick half of the edges of $C_i$ to be in $M$, and there are only two ways of doing it. 
		Hence the individual weights of the edges on the cycles do not play any role. 
		In the new weight function we construct, all edges have weight 0, apart from one on each cycle $C_i$ whose weight we set to the difference between $w_i^+$ and $w_i^-$. 
		
		Then, for any stable matching $M$ of $\SMP G$, we have that $$w(M) = \tilde{w}(M) + W^-,$$ where $$W^- = \sum_{i \in [k]} w_i^- + \sum_{i \in [m]} w(e_i).$$ 
		In particular, $M$ is optimal for $\tilde w$ if and only if it is optimal for $w$.
		
		Now, ovserve that for each agent $v \in V$, at most one of its neighboring edges has non-zero weight in $\tilde{w}$, 
		and this edge would have to be either the most or the least preferred by $V$ in $H$. 
		This implies that the weight function $\tilde{w}$ has the so-called \emph{U-shaped property} of Teo and Sethuraman \cite{teo1998}. 
		Thus we can find, in polynomial time, a stable matching $M^\star$ such that $$\tilde{w}(M^\star) \leq 2 \tilde{w}(M_{OPT})$$ 
		meaning that $$w(M^\star) - W^- \leq 2 (w(M_{OPT}) - W^-)$$ and so in fact $$w(M^\star) \leq 2w(M_{OPT}) - W^- \leq 2w(M_{OPT}).$$ \qed
	\end{proof}
	
\bibliographystyle{splncs03}	
\bibliography{biblio} 	
	
\newpage
	
\section*{Appendix}


\subsubsection{Proof of Lemma \ref{lemma:best_worst_edges_in_EM}.}
		$(\Rightarrow)$: Suppose that $u = f_{\overline H}(v)$. 
		If there was any partner $w$ for $u$ such that $uw \in E_M$ and $v \pref u w$, 
		then in a (perfect) stable matching $M$ containing $uw$, $v$ would have a different partner than $u$, in particular $u \pref v M(v)$ and the edge $uv$ would be blocking.
		Hence $v = l_{\overline H}(u)$. 
		
		$(\Leftarrow)$: Every agent $w$ is the most preferred in $\overline H$ by precisely one agent, 
		as otherwise there would be two agents $x, y$ such that $w = f_{\overline H} (x) = f_{\overline H} (y)$, implying $x = l_{\overline H} (w) = y$. For each agent $w$,
		denote by $w^-$ the unique agent for which $w = f_{\overline H}(w^-)$. Then, for each $w$, $l_{\overline H}(w) = w^-$ by the above implication. 
		In particular, for every agent $w$ we have $x = l_{\overline H} (w) \Leftrightarrow x = w^- \Leftrightarrow w = f_{\overline H}(x)$. \qed
		
\subsubsection{Proof of Lemma \ref{lemma:edges_in_a_SM}.}
Let $\SMP G$ and $e = uv$ be given. There exists a stable matching $M$ of $\SMP G$ s.t. $uv \in M$ if and only if there exists a perfect stable matching $M'$ of $\SMP {G[V']}$ 
not containing any of the edges of $G[V']$ which would make the matching $M' \cup \{e\}$ blocking in $G$. 

If $M'$ is a matching in $G[V']$ containing an edge $wy$ such that $uw \in E$, $w\pref u v$, and $u \pref w y$, then the edge $uw$ is blocking for the matching $M' \cup \{uv\}$. 
Thus any matching $M'$ for which the matching $M' \cup \{uv\}$ is stable in $\SMP G$ excludes all edges in $E_u \cup E_v$. 
On the other hand, any perfect matching $M'$ that is stable in $\SMP {G'}$ is stable in $\SMP {G[V']}$, 
since none of the edges in $E_u \cup E_v$ is blocking for $M'$ in $\SMP {G'}$.

Thus $e$ appears in a stable matching of $\SMP G$ if and only if there exists a perfect stable matching  in $\SMP {G'}$. 
$G'$ can be constructed in polynomial time and it can be decided whether $\SMP {G'}$ contains a perfect stable matching using Irving's algorithm in polynomial time. \qed

\subsubsection{Proof of Lemma \ref{lemma:chceck_removing_edge_no_effect}.}
		Suppose we want to check whether an edge $e = uv \in E \setminus E_M$ can be removed without affecting the set of stable matchings. 
		Let $\psi_u(e) = \{ f \in E: e \pref u f\}$ and $\psi_v(e) =  \{f \in E: e \pref v f\}$ be the sets of edges 
		strictly dominated by the edge $e$ at the endpoints $u$ and $v$, respectively. 
		
		Then a matching $M$ is stable in $\SMP {G \setminus e}$, but not in $\SMP G$ if and only if $M$ is a (perfect) stable matching of $\SMP {G \setminus e}$ 
		such that $|M \cap \psi_v(e)| = |M \cap \psi_u(e)| = 1.$
		
		Since $|\psi_u(e)|,|\psi_v(e)| \leq n$, we have to check at most $\mathcal O(n^2)$ combinations of edges, 
		checking for each pair $(f', f'') \in \psi_u(e) \times \psi_v(e)$ whether it can be extended to a perfect stable matching $M$ of $\SMP {G \setminus e}$. 
		For each such a pair, the check can be carried in polynomial time \cite{dias2003}. \qed

\subsubsection{Proof of Lemma \ref{lemma:worst_edges_redundant}.}
		If there was a matching $M'$ that was stable in $\SMP {G\setminus e}$, but not stable in $\SMP G$, then $e$ would be the only edge blocking $M'$. 
		Since $u = l_G(v)$ and all stable matchings of $\SMP G$ and $\SMP {G \setminus e}$ are perfect, we must have $M'(v) \pref v u$, a contradiction. \qed

\subsubsection{Proof of Lemma \ref{lemma:H_well_defined}.}
		Suppose for contradiction that two executions of the algorithm, $\mathcal T_1$ and $\mathcal T_2$, yield two different graphs $H_1$ and $H_2$, respectively. 
		Suppose that $e \in E(H_1) \setminus E(H_2)$ is an edge present in $H_1$ that is not in $H_2$. 
		Denote $E_0$ the set of edges that got deleted in $\mathcal T_2$, prior to the deletion of $e$. 
		Without loss of generality suppose that $E_0 \cap E(H_1) = \emptyset$, 
		so in other words, all the edges that were deleted in $\mathcal T_2$ prior to $e$ were deleted in $\mathcal T_1$ too. 
		Then, since in $G_0 = G[E \setminus E_0]$, we removed $e$ as the next edge in the $\mathcal T_2$, $e$ is the worst remaining option for one of it's endpoints in $G_0$. 
		Then $e$ is also the worst remaining option for one of its endpoints in $H_1$ as $E(H_1) \subset (E \setminus E_0)$ and so $e$ can be deleted from $H_1$, a contradiction. \qed

\subsubsection{Proof of Theorem \ref{thm:properties_of_H}.}
		\begin{enumerate}[(i)]
			\item This is clear as edges in $E_M$ never get removed from the graph. 
			Since each removal preserves the set of stable matchings, $E_M$ does not change as the algorithm goes on.
			\item Similarly, the edges are eliminated so that the set of stable matchings stays unchanged.
			\item Edges removed in phase one of Irving's algorithm are not in $E_M$. 
			The set of edges removed in phase one is $E_I = \{ uv \in E: l_{G_I}(u) \pref u v\}$, 
			i.e. these are precisely those edges that are worse for one of the endpoints then its worst remaining choice in the resulting graph $G_I$. 
			Consequently, each of the edges in $E_I$ at some point becomes an edge as in Lemma \ref{lemma:worst_edges_redundant}
			in the current remaining subgraph of $G$.
			Thus all the edges that get eliminated in phase one get also eliminated during our procedure. Hence $H \subseteq G_I$. 
			\item If this was not true than phase one applied to $H$  would eliminate at least one edge. 
			There would be at least one eliminated edge that was the worst remaining choice for one of its endpoints.
			\item This is a property of any graph after applying phase one of Irving's algorithm to it.
			\item Since $vf_H(v) = f_H(v)l_H(f_H(v))$, if it was not in $E_M$, it would get eliminated.
			\item An edge in $E_H \setminus E_M$ must be neither the last, thus nor the best remaining option for both of its endpoints.
		\end{enumerate} \qed

\subsubsection{Proof of Lemma \ref{lemma:worst_edges_redundant_inequalities}.}
		We have already proved part $(i)$ as Lemma \ref{lemma:worst_edges_redundant}. For part $(ii)$, it is enough to show that $(\phi(e) \cap E_M) \supseteq (\phi(f) \cap E_M)$ for some edge $f \in G \setminus e$.
		Since $u$ is $v$'s worst remaining choice, we have that $(\delta(v) \cap E_M) \subseteq (\phi(e) \cap E_M)$. 
		In other words, all edges in $E_M$ adjacent to $v$ are contained in $\phi(e) \cap E_M$. 
		
		Let $w$ be such that $w = l_{\overline H} (v)$ and let $f = wv$. Note that such an $f$ must exists and also that $f \neq e$. 
		Since $w = l_{\overline H} (v)$, by Lemma \ref{lemma:best_worst_edges_in_EM} we have that $v = f_{\overline H} (w)$ and thus $(\phi(f) \cap E_M) = (\delta(v) \cap E_M)$. 
		Hence $(\phi(f) \cap E_M) \subseteq (\phi(e) \cap E_M)$ and the inequality $x (\phi(e) \cap E_M) \geq 1$ is dominated by the inequality $x(\phi(f) \cap E_M) \geq 1$. \qed

\subsubsection{Proof of Lemma \ref{lemma:vertices_semistable_partitions}.}
		By Remark \ref{remark}, all the vertices of $\overline{FSM} \SMP G$ are projections of the vertices of $FSM \SMP G$. 
		By Theorem \ref{thm:chen} and by the matching inequality (\ref{ineq:match}), which is satisfied with equality by each $v \in V$ by Lemma \ref{lemma:AR_V0_V1}, 
		we see that, for each $v \in V$, we have either that precisely one of its adjacent edges has value 1 and all the others zero, 
		or precisely two of its adjacent edges have value $\frac{1}{2}$ and all the others have value zero. 
		This immediately gives us that $supp(x)$ consists of single edges and vertex-disjoint cycles.
		
		Suppose that $C$ is such a a cycle and w.l.o.g suppose $C = \{v_1,\dots, v_m\}$. 
		Then, since the stability inequality (\ref{ineq:stab}) must be satisfied at the edge $v_1v_2$ and also by Lemma \ref{lemma:AR_xe_positive}, 
		we have that precisely one of the edges $v_2v_3$ and $v_1v_m$ dominates $v_1v_2$ and one is dominated by it. 
		Say, w.l.o.g., that $v_3 \pref {v_2} v_1$ and $v_2 \pref {v_1} v_m$. 
		Then, since $v_3 \pref {v_2} v_1$, Lemma \ref{lemma:AR_xe_positive} applied to edge $v_2v_3$ implies that $v_4 \pref {v_3} v_2$. 
		Inductively, we see that $v_{i+1} \pref {v_i} v_{i-1}$ for all $i = 1,\dots, m$. \qed

\end{document}